\documentclass[letterpaper, 10 pt, conference]{ieeeconf}  % Comment this line out if you need a4paper

\IEEEoverridecommandlockouts                              % This command is only needed if 
\usepackage{graphics} % for pdf, bitmapped graphics files
\usepackage{epsfig} % for postscript graphics files
\usepackage{mathtools}
\usepackage{times} % assumes new font selection scheme installed
\usepackage{amsmath} % assumes amsmath package installed
\usepackage{amssymb}  % assumes amsmath package installed
\usepackage{float}
\usepackage{color,soul}
\usepackage[ruled,vlined]{algorithm2e}
\usepackage{bbm}
\usepackage{threeparttable}

\usepackage{xargs}
\usepackage{cite}
\newcommandx{\iman}[2][1=]{\todo[linecolor=blue,backgroundcolor=blue!25,bordercolor=blue,author=Iman,#1]{#2}}

\usepackage{geometry}
\title{\LARGE \bf
Generating Minimum-Snap Quadrotor Trajectories Really Fast
}

\author{Declan Burke, Airlie Chapman, Iman Shames% <-this % stops a space
\thanks{The first and second authors are with the Department of Mechanical Engineering and third with the Department of Electrical and Electronic Engineering, The University of Melbourne, VIC, 3010 Australia. Emails: {\tt\small \{ declanb@student., airlie.chapman@, iman.shames@\}unimelb.edu.au.}
This work is partially supported by Defence Science and Technology Group, through agreement MyIP: ID9156 entitled ``Verifiable Hierarchical Sensing, Planning and Control'', the Australian Government, via grant AUSMURIB000001 associated with ONR MURI grant N00014-19-1-2571.}}

\begin{document}

\newtheorem{lemm}{\bf Lemma}
\renewcommand{\thelemm}{\arabic{lemm}}
\newtheorem{prop}[lemm]{\bf Proposition}
\renewcommand{\theprop}{\arabic{lemm}}
\newtheorem{prob}[lemm]{\bf Problem}
\renewcommand{\theprob}{\arabic{lemm}}
\newtheorem{algo}[lemm]{\bf Algorithm}
\renewcommand{\thealgo}{\arabic{lemm}}
\newtheorem{rema}[lemm]{\bf Remark}
\renewcommand{\therema}{\arabic{lemm}}

\renewcommand{\thetable}{\arabic{table}}  

\maketitle
\thispagestyle{empty}
\pagestyle{empty}

%%%%%%%%%%%%%%%%%%%%%%%%%%%%%%%%%%%%%%%%%%%%%%%%%%%%%%%%%%%%%%%%%%%%%%%%%%%%%%%%
\begin{abstract}
We propose an algorithm for generating minimum-snap trajectories for quadrotors with linear computational complexity with respect to the number of segments in the spline trajectory. Our algorithm is numerically stable for large numbers of segments and is able to generate trajectories of more than $500,000$ segments. The computational speed and numerical stability of our algorithm makes it suitable for real-time generation of very large scale trajectories. We demonstrate the performance of our algorithm and compare it to existing methods, in which it is both faster and able to calculate larger trajectories than state-of-the-art. We also show the feasibility of the trajectories experimentally with a long quadrotor flight.
\end{abstract}

%%%%%%%%%%%%%%%%%%%%%%%%%%%%%%%%%%%%%%%%%%%%%%%%%%%%%%%%%%%%%%%%%%%%%%%%%%%%%%%%
\section{INTRODUCTION}
Quadrotor trajectory optimization for has been extensively studied with a variety of approaches suggested \cite{Gao2019}. Mellinger and Kumar \cite{Mellinger2011} pioneered minimum-snap trajectory generation for quadrotors. They used smooth trajectories to control quadrotors and formulated a quadratic program (QP) to calculate continuous splines. A method of solving this QP was proposed by Richter {\it et al.}~\cite{Richter2016}, who also included a procedure for making trajectories safer. Mellinger {\it et al.}~\cite{Mellinger2012} and Deits and Tedrake \cite{Deits2015} further developed the use of minimum-snap splines and incorporated collision avoidance by formulating mixed-integer programs to generate trajectories. Another approach to trajectory generation is in the use of motion primitives generated by solving optimal control problems, often to minimize the snap of the primitive. Lui {\it et al.}~\cite{Liu2017} discretize the state space using primitives and then perform a graph search, while Muller {\it et al.}~\cite{Mueller2015} evaluate many primitives at each controller update step and execute the feasible primitive with the lowest cost.

% In this paper, we revisit the equality constrained, QP formulation of \cite{Mellinger2011, Richter2016} used to generate a minimum snap trajectory\footnote{The authors note that an inequality constrained QP should instead be employed to introduce position, velocity and acceleration constraints useful for obstacle avoidance. See \cite{de2017new} for a problem formulation and further details.}. The aforementioned QP formulation leads to conceptually simple optimisation problems. However, solving the resultant optimization program can be challenging in practice. Richter {\it et al}. \cite{Richter2016} note numerical instability in using constrained optimisation methods to solve for trajectories of more than four segments. Their solution allows for calculating trajectories of more than $50$ segments that exhibits better empirical numerical stability than other existing methods. However, no analysis for this apparent improvement in numerical stability is provided beyond the presented numerical examples. Further, other properties of solutions to minimum-snap trajectory generation problems such as their computational complexity remain unexplored. This solution does not scale well with the number of segments in the spline and numerically becomes unstable for trajectories of long flight times.

In this paper, we revisit the equality constrained QP formulation of \cite{Mellinger2011, Richter2016} used to generate a minimum snap trajectory\footnote{The authors note that an inequality constrained QP should instead be employed to introduce position, velocity and acceleration constraints useful for obstacle avoidance. See \cite{de2017new} for a problem formulation and further details.}. The aforementioned QP formulation leads to conceptually simple optimisation problems. However, solving the resultant optimization program can be challenging in practice. Richter {\it et al}. \cite{Richter2016} note numerical instability in using constrained optimisation methods to solve for trajectories of more than four segments. De Almeida and Akella further investigate the numerical instability of the problem formulation and identify the ill-conditioned matrix that is the most significant source of error \cite{de2017new}. Beyond numerical stability, other properties of solutions to minimum-snap trajectory generation problems such as their computational complexity remain unexplored.

\begin{figure}[tp]
	\centering
	{\includegraphics[width=0.98\linewidth]{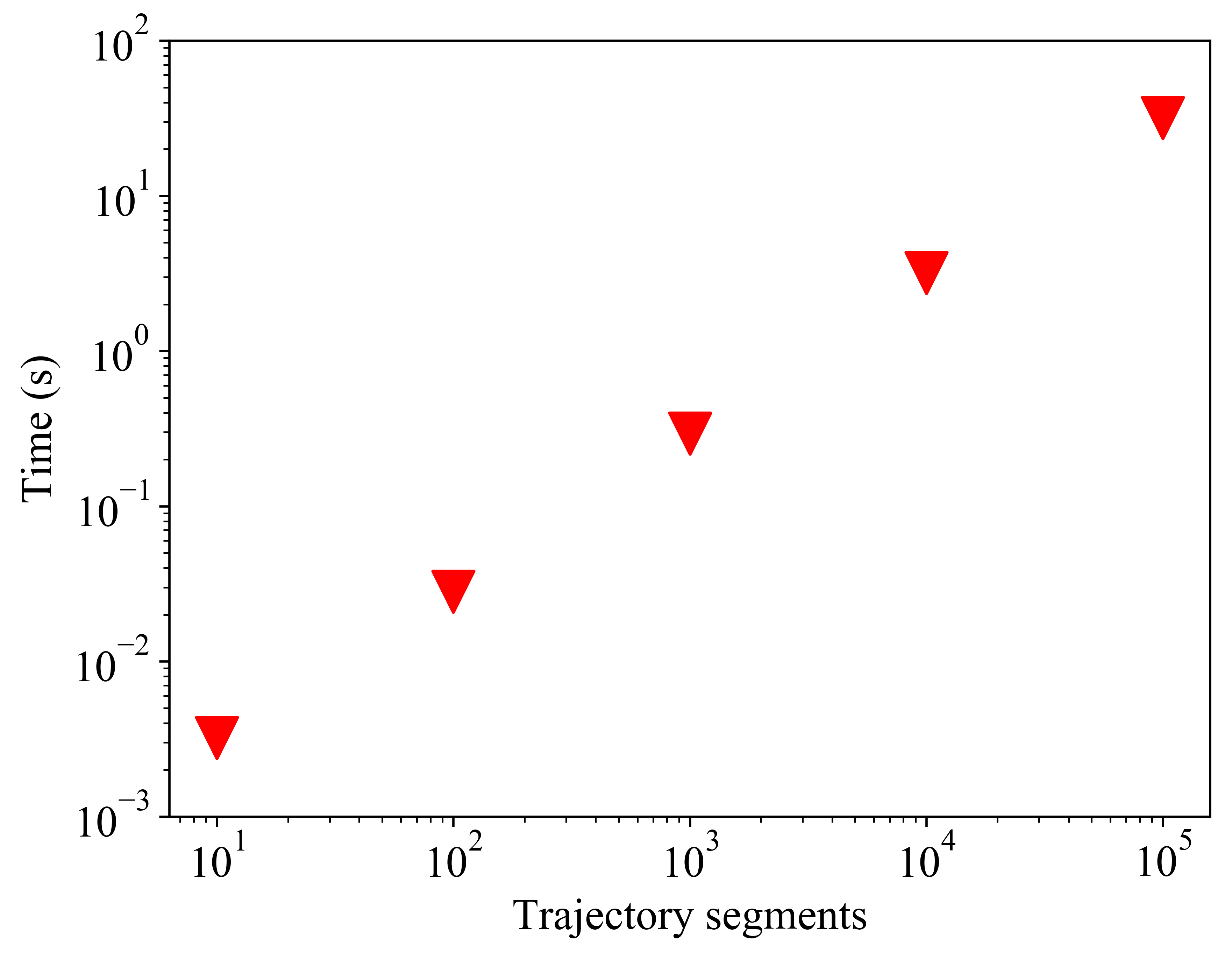}}
	\caption{The computation time required to generate minimum-snap trajectories for splines of a range of polynomial segments using the proposed algorithm (Algorithm \ref{algo1}).}
	\label{lin_comp} 
\end{figure}

In this paper, we propose a method that solves the QP problem corresponding to a minimum snap trajectory generation with a very large number of segments for long, albeit finite, flight times. We achieve this by explicitly identifying and exploiting the algebraic structure of the problem. Particularly, first, we show that our solution is of linear computational complexity in the number of segments of the trajectory. Second, inspired by the properties of confluent Vandermonde matrices we introduce a new change of variable to arrive at an equivalent yet better conditioned problem. Our work differs from that of \cite{de2017new} as we address the numerical instability of the original QP formulation of \cite{Mellinger2011}, \cite{Richter2016} rather than proposing a new model for minimising the snap of the trajectory. 

In summary, our algorithm solves problems many times larger than state-of-the-art methods, demonstrating with trajectories of more than $500,000$ segments (see Fig. \ref{lin_comp}). 

The paper's outline is as follows. In the next section, we introduce notations along with the minimum snap problem. In Section \ref{sec::structure}, the method for solving the minimum snap problem is described. In Section \ref{sec::condition}, we study the conditioning problems inherent to the vanilla formulation of the minimum snap problem and propose a change of variable that alleviate this issue. Experimental results are presented in Section \ref{sec::experiments}. Concluding remarks come in the end.

% We also include a discussion, comparing this to existing methods.

% Our second contribution is a reformulation of the original program to improve the conditioning of the problem. We identify the source of numerical instability in current methods and recast the QP to a nondimensional form to alleviate this issue. Under this reformulation, 

\section{THE MINIMUM SNAP PROBLEM}

\subsection{Notation}
    % We parameterize polynomials with respect to time $t$. When discretizing time, we consider only nondecreasing times represented by the vector $T=[t_0,\ldots,t_k]^T\in\mathbb{R}^k$, such that $t_{j-1}\leq t_j$ for $j=1,\ldots,k$. 

    % We represent polynomials $\pi_i(t;p_i):\mathbb{R}\to \mathbb{R}$ in terms of the vectors of coefficients $p_i=[p_{0,i},\ldots,p_{n,i}]^T$. For equations involving many polynomials $p_i$, for $i=1,\ldots,k$, we stack the vectors as $p=[p_1^T,\ldots,p_k^T]^T\in\mathbb{R}^{kn}$.
    
    % Derivatives are represented as scalar values sampled at particular times, where the $j$ derivative of the spline $\pi_i(t;p_i)$ sampled at its endpoints at time $t$ is $d_{i,j,t}\in\mathbb{R}$. They are also represented as vectors of the first $s$ derivatives $d_{i,t}=[d_{i,0,t}^T,\ldots,d_{i,s,t}^T]^T\in\mathbb{R}^{2s}$. For equations involving many derivative values $d_{i,t}$ at $t_{i-1}$ and $t_i$, for $i=1,\ldots,k$, we write $d=[d_{1,t_0}^T,d_{1,t_1}^T,\ldots,d_{k,t_{k-1}}^T,d_{k,t_k}^T]^T\in\mathbb{R}^{2sk}$.

    We use the following pieces of notation throughout the paper. We write the dot product as $\cdot$, and the identity matrix as $I_s\in\mathbb{R}^{s\times s}$. The condition number of a matrix $V$ is defined $\kappa(V)=\lVert V\rVert \lVert V^{-1}\rVert$. The cardinality of a set $\Xi$ is written as $|\Xi|$. % We abbreviate $i=1,\ldots,k$ with $i\in[k]$; \hl{how to notate decreasing order}. 
    A function $\pi$ of $t$ that is parameterised by $p$ is written $\pi(t;p)$. We use the order notation for the sequences $\{\eta_i\}$ and $\{\nu_i\}$, to write $\eta_i = O(\nu_i)$ if there exists a positive constant $C\in\mathbb{R}$ such that $|\eta_i|\leq C|\nu_i|$ for $i$ sufficiently large. Given a sequence of column vectors $\{x_i\}$ and matrices $\{V_i\}$ for $i=1,\ldots,k$, let
    \begin{align*}
        \mathrm{vec}(\{x_i\})\coloneqq\begin{bmatrix}
        x_1 \\ \vdots \\ x_k
        \end{bmatrix},~\mathrm{diag}(\{V_i\})\coloneqq\begin{bmatrix}
        V_1 &  &  \\ & \ddots & \\ & & V_k
        \end{bmatrix}.
    \end{align*}
    
    \subsection{Formulating an Optimization Program}
    Mellinger and Kumar established that a smooth trajectory in terms of position $[x,y,z]^T\in\mathbb{R}^3$ and yaw $\psi\in SO(2)$ can be used to calculate motor commands for a quadrotor, i.e., the system's control inputs \cite{Mellinger2011}. Their method of generating this trajectory involves formulating and solving four QPs. The QPs are of the same form and they may be solved separately, so we consider the general problem of generating a spline $\pi(t)$ that can represent any of $x(t),y(t),z(t)$ or $\psi(t)$\footnote{The minimum snap objective applies only to $\pi_i(t;p_i)=x_i(t),~y_i(t)$ and $z_i(t)$. For $\pi_i(t;p_i)=\psi_i(t)$, the integrand in objective is $[\frac{d^2p_i}{dt^2}(t)]^2$ (see \cite{Mellinger2011}).}. % For further details on quadrotor dynamics and controlling to minimum-snap trajectories, we suggest the reader consult the original paper \cite{Mellinger2011}.
    
    We now introduce the representation of the system used to formulate minimum-snap trajectory generation as an optimization program. We consider nondecreasing times as $T=[t_0,\ldots,t_k]^T\in\mathbb{R}^{k+1}$ such that $t_{i-1}\leq t_i$ for $i=1,\ldots,k$. The conventional choice for the trajectory is a continuous spline $\pi(t;p):[t_0,t_k]\to \mathbb{R}$ of $k\geq1$, $k$ integer, segments
    \begin{align*}
        \pi(t;p)&=\begin{cases}
        \pi_1(t;p_i) & t_0\leq t<t_1, \\
        \quad \vdots & \\
        \pi_k(t;p_k) & t_{k-1}\leq t<t_k,
        \end{cases}
    \end{align*}
    where each segment $\pi_i(t;p_i):[t_{i-1},t_i]\to \mathbb{R}$ is a polynomial of order $n-1$, where $n$ is a positive integer. We represent the segments as $\pi_i(t;p_i)=[1,t,\ldots,t^{n-1}] p_i$, where $p_i=[p_{0,i},\ldots,p_{n-1,i}]^T$ for all $i=1,\ldots,k$, and the entire spline as $p=\mathrm{vec}(\{p_i\})\in\mathbb{R}^{kn},~i=1,\ldots,k$.
    
    Constraints on the trajectory $\pi(t;p)$ are introduced by setting the value of the segments $\pi_i(t;p_i)$ and their derivatives at $t_{i-1}$ and $t_i$, $i=1,\ldots,k$. The value of $j$th derivative of the segment $\pi_i(t;p_i)$ at time $t$ is the scalar $\delta_{i,j,t}$, and we consider up to the $(s-1)$th derivative, where $s$ is a positive integer. We represent the first $s-1$ derivatives as the vector $d_{i,t}=\mathrm{vec}(\{\delta_{i,j,t}\})\in\mathbb{R}^{s},~j=0,\ldots,s-1$, and for many derivative values $d_{i,t}$ at $t_{i-1}$ and $t_i$, % for $i=1,\ldots,k$, 
    we write the vector $d=\mathrm{vec}(\{[d_{i,t_{i-1}}^T,d_{i,t_i}^T]^T\})\in\mathbb{R}^{2sk}$ for $i=1,\ldots,k$. We require continuity between segments, so we equate $d_{i,t_i}=d_{i+1,t_i}$ for all $i=1,\ldots,k-1$. Further constraints can be introduced with the set $\Xi=\{(i,j,t)\}$, if $\kappa\in\Xi$ then $\delta_\kappa=\beta_\kappa$ for some prescribed derivative $\beta_\kappa\in\mathbb{R}$. There is a total of $|\Xi|=2sk-m$ constraints defined by $\Xi$, where $m=\sum_{i=1}^{k}m_{i,t_{i-1}}+m_{i,t_i}$ and $m_{i,t_{i-1}}$ and $m_{i,t_i}$ are the number of derivative values of the polynomial $\pi_i(t;p_i)$ without specified values at $t_{i-1}$ and $t_i$, respectively. % Note, $m_{i,t}\leq s$ for $i =1,\ldots,k,~t=t_{i-1},t_i$.

    % With this representation of polynomials and their derivatives, we formally introduce the minimum-snap trajectory generation problem.

% \begin{prob}\label{prog1}

As shown in \cite{Mellinger2012,Richter2016}, the minimum-snap trajectory $\pi(t;p^\star)$ can be generated via solving the following QP for the minimizing arguments $(p^\star,d^\star)$ of
\begin{subequations} \label{form1}
    \begin{align}
        \arg\min_{p, d} \quad &p^T Q(T) p,  \label{form1:cost} \\
        \text{subject to} \quad & A(T)p=d, \label{form1:con1} \\
        & \Lambda d =0,  \label{form1:con2} \\
        & \Pi d = \beta, \label{form1:con3}
    \end{align}
\end{subequations}
where $\beta\in\mathbb{R}^{2sk-m}$, the vector of $\beta_\kappa$ for all $\kappa\in\Xi$, and the vector of times $T=[t_0,\ldots,t_k]\in\mathbb{R}^{k+1}$. The block diagonal matrices $Q(T)=\mathrm{diag}\{Q_1(T),\ldots,Q_k(T)\}\in\mathbb{R}^{nk\times nk}$ and $A(T)=\mathrm{diag}\{A_1(T),\ldots,A_k(T)\}\in\mathbb{R}^{2sk\times nk}$ are constructed from the elements $Q_i(T)\in\mathbb{R}^{n\times n}$ and $A_i(T)\in\mathbb{R}^{2s\times n}$. The $i$th matrix $A_i(T)$ is
% \begin{align*}
%     A_i(T) &=\begin{bmatrix}
%     1 & t_{i-1} & t_{i-1}^2 & \ldots & t_{i-1}^{n-1} \\
%     0 & 1 & 2t_{i-1} & \ldots & (n-1)t_{i-1}^{n-2} \\
%     0 & 0 & 2 & \ldots & (n-1)(n-2)t_{i-1}^{n-3} \\
%     \vdots & \vdots & \vdots & \ddots &  \vdots \\ 
%     0 & 0 & 0 & \ldots & \frac{(n-1)!}{(n - s)!}t_{i-1}^{n-s} \\
%     1 & t_i & t_i^2 & \ldots &  t_i^{n-1} \\
%      \vdots & \vdots & \vdots & \ddots &  \vdots \\ 
%     0 & 0 & 0 & \ldots & \frac{(n-1)!}{(n - s)!}t_i^{n-s} \\
%     \end{bmatrix}.
% \end{align*}
\begin{align*}
    A_i(T) &=\begin{bmatrix}
    1 & t_{i-1} & t_{i-1}^2 & \ldots & t_{i-1}^{n-1} \\
    0 & 1 & 2t_{i-1} & \ldots & (n-1)t_{i-1}^{n-2} \\
    \vdots & \vdots & \vdots & \ddots &  \vdots \\ 
    0 & 0 & 0 & \ldots & \frac{(n-1)!}{(n - s)!}t_{i-1}^{n-s} \\
    1 & t_i & t_i^2 & \ldots &  t_i^{n-1} \\
     \vdots & \vdots & \vdots & \ddots &  \vdots \\ 
    0 & 0 & 0 & \ldots & \frac{(n-1)!}{(n - s)!}t_i^{n-s} \\
    \end{bmatrix}.
\end{align*}
For details regarding the construction of $Q_i(T)$, see Appendix \ref{sec:const}. The $i$th matrix $A_i(T)$ is the transpose of a confluent Vandermonde matrix, which we will revisit in detail in Section \ref{sec::condition}. A pertinent property as a confluent Vandermonde matrix is that the matrix $A_i(T)$ is nonsingular for distinct $t_i$ \cite{higham2002accuracy}.

Let $\omega_{i,t} =\{\kappa \in \Xi | i\in \Xi,~t\in\Xi\}$. The selection matrix $\Pi=\mathrm{diag}(\{\Pi_{1,t_0},\Pi_{1,t_1},\ldots, \Pi_{k,t_{k-1}}, \Pi_{k,t_k}\})\in\mathbb{R}^{(2sk-m)\times 2sk}$ is constructed from the elements $\Pi_{i,t}\in\mathbb{R}^{(s-m_{i,t})\times s}$ such that $[\Pi_{i,t}]_{h,j}=1$ for all  $\eta_v\in\omega_{i,t}$ and $h=1,\ldots,|\omega_{i,t}|$. The coupling matrix $\Lambda \in\mathbb{R}^{s(k-1)\times 2sk}$ is
\begin{align*}
    \Lambda= \begin{bmatrix}
            0 & I_s & -I_s & 0 & 0 & \hdots & 0 & 0 & 0\\
            0 & 0 & 0 & I_s & -I_s & \hdots & 0 & 0 & 0   \\
            &  & \vdots & & & \ddots & & \vdots & \\
            0 & 0 & 0 & 0 & 0 & \hdots & I_s & -I_s & 0 
            \end{bmatrix}.
\end{align*}
The decision variables in \eqref{form1} are the vectors of polynomial coefficients $p\in\mathbb{R}^{nk}$ and vectors of derivative values $d\in\mathbb{R}^{2sk}$. For brevity we will stop explicitly stating time parameterisation of matrices and will write $A$ and $Q$ instead of $A(T)$ and $Q(T)$.

% \begin{rema}
% The $i$th matrix $A_i(T)$ is the transpose of a confluent Vandermonde matrix, which we will revisit in detail in Section \ref{sec::condition}. As a confluent Vandermonde matrix, the matrix $A_i(T)$ is nonsingular for distinct $t_i$ \cite{higham2002accuracy}.
% \end{rema}

\begin{prob}\label{prog1}
Generate a minimum snap spline $\pi(t;p^\star)$ with $k$ polynomial segments of order $n-1$ where $p^\star$ is obtained by solving \eqref{form1} using an algorithm with computational complexity of $O(kn^3)$. 
\end{prob}
% \vspace{1em}
In the next section, we propose a solution to Problem~\ref{prog1}. Later, in Section~\ref{sec::condition}, we further introduce a modification that considerably improves the numerical stability of the proposed solution. 

% \begin{algo}
% The following algorithm solves Program \ref{prog1}.
% \end{algo}

\section{A STRUCTURED SOLUTION}\label{sec::structure}
   \begin{algorithm}[t]
        \For{i=1,\ldots, k} {
            construct $A_i,Q_i,\Sigma_i,b_{i,t_{i-1}}, b_{i,t_i}$; \\
            partition $B_i,C_i,D_i,g_{i,t_{i-1}}, g_{i,t_i}$; \\
            $ \overline{B}_i {\gets\begin{cases}
            B_1 & i=1, \\
            B_{i}+D_{i-1}-C_{i-1}^TB_{i-1}^{-1}C_{i-1} & i=2,\ldots,k;
            \end{cases}}$ 
            \newline $\overline{g}_{i,t_{i-1}} {\gets \begin{cases}
            {g}_{1,t_{i-1}} & i=1, \\
            \begin{aligned}
            & g_{i,t_{i-1}} + g_{i-1,t_{i-2}} \\
            & \quad -C_{i-1}^TB_{i-1}^{-1}\overline{g}_{i-1,t_{i-2}}
            \end{aligned}
             & i=2,\ldots,k;
            \end{cases}}$
        }
        \For{i=k,\ldots, 1} {
            $f_{i,t_{i-1}}\gets  \overline{B}^{-1}_i(\overline{g}_{i,t_{i-1}}-C_if_{i,t_{i}}) $; \\
            $f_{i,t_i} {\gets \begin{cases}
            \begin{aligned}
            & (D_k-C_k^T\overline{B}_k^{-1}C_k)^{-1}  \\
            & \quad (g_{k,t_{k}}-C_k^T\overline{B}_k^{-1}\overline{g}_{k,t_{k-1}})
            \end{aligned} 
            & i=k, \\
            f_{i+1,t_{i-1}} & i=k-1,\ldots,1;
            \end{cases}}$ \\
            solve for $d_{i,t_{i-1}}$ and $d_{i,t_i}$; \\
            % Can use divided differences to solve for $p_i$ here.
            solve for $p_i$;
        }
    
     \caption{Algorithm of Linear Computational Complexity for Solving Problem \ref{prog1}.} \label{algo1}
     \end{algorithm}

\subsection{The Algorithm}
We begin by presenting our algorithm for solving Problem \ref{prog1}. We start by introducing the variables required by the algorithm's calculations. Let
\begin{align}\label{proj1}
    d &= b + \Sigma f,
\end{align}
where $b\in \mathbb{R}^{2sk}$, $f\in \mathbb{R}^{m}$  and $\Sigma\in \mathbb{R}^{2sk \times m}$ is such that $\Pi b=\beta$ and  $\Pi\Sigma=0$. We note that $b$ is constant in the sense its entries are either the parameters $\beta_k$, for $k\in\Xi$, or zeros. We partition $b,~f$ and $\Sigma$ into vectors and matrices of size  $b_{i,t}\in \mathbb{R}^{s}$, $f_{i,t}\in \mathbb{R}^{m_{i,t}}$ and $\Sigma_{i,t}\in \mathbb{R}^{s \times m_{i,t}}$, for $i=1,\ldots,k$ and $t=t_{i-1},t_i$, such that $b=\mathrm{vec}(\{[b_{i,t_{i-1}}^T,b_{i,t_i}^T]^T\})\in\mathbb{R}^{2sk}$ and $f=\mathrm{vec}(\{[f_{i,t_{i-1}}^T,f_{i,t_i}^T]^T\})\in\mathbb{R}^{m}$ for $i=1,\ldots,k$. Further, let $\Sigma_i=\mathrm{diag}(\{\Sigma_{i,t_{i-1}},\Sigma_{i,t_i}\})$ such that   $\Sigma=\mathrm{diag}(\{\Sigma_{i}\})\in\mathbb{R}^{2sk\times m}$ for $i=1,\ldots,k$.

For $i=1,\ldots,k$, let
\begin{subequations}\label{partition}
    \begin{align}
    \Sigma_i^TA_i^{-T}Q_iA_i^{-1}\Sigma_i &= \begin{bmatrix}
    B_i & C_i \\
    C_i^T & D_i
    \end{bmatrix}, \\
    \Sigma_i^TA_i^{-T}Q_iA_i^{-1} \begin{bmatrix}
    b_{i,t_{i-1}} \\
    b_{i,t_i}
    \end{bmatrix}&= \begin{bmatrix}
    g_{i,t_{i-1}} \\
    g_{i,t_i}
    \end{bmatrix}
\end{align}
\end{subequations}
where $B_i\in\mathbb{R}^{m_{i,t_{i-1}}\times m_{i,t_{i-1}}}$, $C_i\in\mathbb{R}^{m_{i,t_{i-1}}\times m_{i,t_i}}$, $D_i\in\mathbb{R}^{m_{i,t_i}\times m_{i,t_i}}$ and $g_{i,t}\in\mathbb{R}^{m_{i,t}}$ for $t\in \{t_{i-1},t_i\}$. 
\begin{prop}\label{prop:kn3}
Algorithm \ref{algo1} solves Problem \ref{prog1} in $O(kn^3)$. 
\end{prop}
\begin{proof}
See Appendix \ref{sec:proof}.
\end{proof}

% We postpone the proof of Proposition~\ref{prop:kn3} to the end of this section.  % Cast as \eqref{form1}, the program can be solved using a range of constrained optimisation methods, including interior point methods. However, noting its poor performance for trajectories of more than four segements \cite{Richter2016}, we strive for a different approach.

% To this end, for $i=1,\ldots,k$ and $t=t_{i-1},t_i$, let  
% \begin{align}\label{proj1}
%     d_{i,t} &= b_{i,t} + \Sigma_{i,t} f_{i,t},
% \end{align}
% where $f_{i,t}\in \mathbb{R}^{m_{i,t}}$, $\Sigma_{i,t}\in \mathbb{R}^{s \times m_{i,t}}$ and $b_{i,t}\in \mathbb{R}^{s}$ are such that, for $\Sigma=\mathrm{diag}(\{\Sigma_{1,t_0},\Sigma_{1,t_1},\ldots, \Sigma_{k,t_{k-1}}, \Sigma_{k,t_k}\})\in\mathbb{R}^{kn\times m}$ and $b=\mathrm{vec}(\{b_{1,t_0},b_{1,t_1},\ldots, b_{k,t_{k-1}}, b_{k,t_k}\})\in\mathbb{R}^{2sk}$, $\Pi b=\beta$ and $\Pi\Sigma=0$.
    
    % Note, in the solution to $Ux=y$, the values of the Lagrangian multipliers $\lambda_i$ are set by

\begin{rema}
    We note that in practice $s$ is usually chosen to be $5$, such that continuity is enforced up to the snap of position \cite{Mellinger2011}. The number of segments $k$ however depends on the size of the trajectory and can be very large. This is then an improvement over other methods for solving Problem \ref{prog1}, which typically have a cubic computation complexity in $k$. The upper bound on the size of the equations \eqref{KKT} is $kn+(k-1)s$. Solving with Gaussian Elimination would then require $O(\frac{2}{3}k^3(n^3+s^3))$ operations. % scaling poorly in the number of splines $k$, their order $n$ and the degree of continuity enforced $s$.
    The best current method, proposed by Richter {\it et al.} \cite{Richter2016}, includes the inversion of a matrix of size $m$, and thereby needs $O(2k^3n^3)$ operations. % In practical implementation, choosing $s=4$ and $n=9$ ensures continuity to the fourth derivative and provides enough degrees of freedom to find feasible solutions. To the contrary, the number of splines $k$ is unbounded, so, in terms of computation complexity, improving scaling with respect to $k$ is the most important aspect of a Minimum Snap solution. 
    Linear in $k$, our approach dramatically reduces the computation time required to solve Problem \ref{prog1}, particularly for large trajectories.
    \end{rema}

\section{IMPROVING CONDITIONING}\label{sec::condition}

    We now present a reformulation that results in a QP with well-conditioned matrices. % This reformulation along with Algorithm \ref{algo1} result in a methodology for solving Problem \ref{prog1} for trajectories with many segments (see Remark \ref{rem:synergy} for more information). 
    To motivate this reformulation, we take a short detour to explore Vandermonde matrices. The $A_i$ are transposed confluent Vandermonde matrices and are notoriously ill-conditioned \cite{higham2002accuracy}. These matrices are ultimately the culprits behind the numerical instability of \eqref{form1}.
    
    A Vandermonde matrix is defined by the scalars $z_0,z_1,\ldots,z_n\in \mathbb{C}$ as
    \begin{align*}
        V(z_0,z_1,\ldots,z_n)=\begin{bmatrix}
        1 & 1 & \ldots & 1 \\
        z_0 & z_1 & \ldots & z_n \\
        \vdots & \vdots & & \vdots \\ 
        z_0^n & z_1^n & \ldots &  z_n^n
        \end{bmatrix}\in \mathbb{C}^{(n+1)\times(n+1)}.
    \end{align*}
    One way of generalising the standard Vandermonde matrix is to allow confluency, that is, including columns with elements that are the differentiated elements of other columns. % An example of a confluent Vandermonde matrix (with real elements) is the matrix $A_i^T$ of \eqref{form1}.
    %% Original presentation
    % \begin{align*} \label{eq::A_vand}
    %     A_i^T &=\begin{bmatrix}
    %     1 & 0 & \ldots & 0 & 1 & \ldots & 0 \\
    %     t_{i-1} & 1 & \ldots & 0 & t_i & \ldots & 0 \\
    %     t_{i-1}^2 & 2t_{i-1} & \ldots & 0 & t_i^2 & \ldots & 0 \\
    %     \vdots & \vdots & & \vdots &  \vdots & & \vdots \\ 
    %     t_{i-1}^n & nt_{i-1}^{n-1} & \ldots & \frac{n!}{s!}t_{i-1}^{n-s} & t_i^n & \ldots & \frac{n!}{s!}t_i^{n-s}
    %     \end{bmatrix}.
    % \end{align*}
    % \begin{align} \label{eq::A_vand}
    %     A_i^T &=\begin{bmatrix}
    %     1 & 0 & \ldots & 0 & 1 & \ldots \\
    %     t_{i-1} & 1 & \ldots & 0 & t_i & \ldots \\
    %     t_{i-1}^2 & 2t_{i-1} & \ldots & 0 & t_i^2 & \ldots \\
    %     \vdots & \vdots & & \vdots &  \vdots & \\ 
    %     t_{i-1}^n & nt_{i-1}^{n-1} & \ldots & \frac{n!}{s!}t_{i-1}^{n-s} & t_i^n & \ldots
    %     \end{bmatrix}.
    % \end{align}
    % \begin{align} \label{eq::A_vand}
    %     A_i^T &=\begin{bmatrix}
    %     1 &\ldots & 0 & 1 & \ldots & 0 \\
    %     t_{i-1} & \ldots & 0 & t_i & \ldots & 0 \\
    %     t_{i-1}^2 & \ldots & 0 & t_i^2 & \ldots & 0 \\
    %     \vdots & & \vdots &  \vdots & & \vdots \\ 
    %     t_{i-1}^n & \ldots & \frac{n!}{s!}t_{i-1}^{n-s} & t_i^n & \ldots & \frac{n!}{s!}t_i^{n-s}
    %     \end{bmatrix}.
    % \end{align}
    One strategy to improve the conditioning of $V$ is to carefully select the points $z_j$, for $j=0,\ldots,n$. Ideally, $V$ is perfectly conditioned, i.e., $\kappa(V)=1$, which occurs when the points are roots of unity $z_j=\mathrm{exp}(\frac{2\pi \bf{i}}{n}j)$ (here $\bf{i}$ is the imaginary unit), for $j=0,\ldots,n-1$ \cite{Gautschi1974}. Taking the points of a confluent Vandermonde matrix to be roots of unity has also been shown to improve its condition number \cite{gautschi1978inverses}.
    
    % In the context of Problem \ref{prog1}, we cannot improve the condition number of the $A_i$ by choosing complex points as $t_i\in\mathbb{R}$ for $i=0,\ldots,k$. We can however scale the problem such that all the $A_i$ and $Q_i$ have the nondimensional points $t_{i-1}=-1$ and $t_i=1$ (with associated matrices $A_{\pm 1}$ and $Q_{\pm 1}$). The following program, which is similar to \eqref{form1}, uses the matrices $A_{\pm 1}$ and $Q_{\pm 1}$. We will show the solutions to \eqref{form3} can be used to solve Problem \ref{prog1}, and hence we may use a program with better conditioned matrices. The nondimensional program is
    To formulate a program with better conditioned matrices than \eqref{form1}, we thus scale the problem such that all the $A_i$ and $Q_i$ have the nondimensional points $t_{i-1}=-1$ and $t_i=1$ (with associated matrices $A_{\pm 1}$ and $Q_{\pm 1}$). The nondimensional program is
    \begin{subequations} \label{form3}
        \begin{align}
            \min_{\overline{p},d} \quad &\overline{p}^T\overline{Q}\overline{p}, \label{form3:cost} \\
            \text{s.t.} \quad & \overline{A}\overline{p}=d, \label{form3:con1}\\
            & \Lambda d =0, \label{form3:con2}\\
            & \Pi d = \beta,\label{form3:con3}
        \end{align}
    \end{subequations}
    where $\overline{A}=\mathrm{diag}\{\Gamma_1 A_{\pm1},\ldots,\Gamma_k A_{\pm1}\}\in\mathbb{R}^{2sk\times nk}$ and $\overline{Q}=\mathrm{diag}\{Q_{\pm1},\ldots,Q_{\pm1}\}\in\mathbb{R}^{nk\times nk}$ are constructed from elements $\Gamma_i A_{\pm1}\in\mathbb{R}^{n\times n}$ and $Q_{\pm1}\in\mathbb{R}^{n\times n}$. The matrix $\Gamma_i=\mathrm{diag}\{\Delta_i^{0}, \Delta_i^{-1}, \ldots, \Delta_i^{-(s-1)}, \Delta_i^0, \ldots, \Delta_i^{-(s-1)}\}\in \mathbb{R}^{n\times n}$ scales the rows of the $A_{\pm 1}$
    with elements $\Delta_i=(t_i-t_{i-1})/{2}$ for $i=1,\ldots,k$.
    
    The following proposition is needed to show when \eqref{form3} can be used to find a solution to Problem \ref{prog1}. It considers the program for a minimum snap trajectory of a single segment in both a dimensional ($t_{i-1}\leq t <t_i$) and nondimensional form ($-1\leq t <1$). % We provide conditions for when their solutions are the same polynomials, translated and scaled. 
    
    \begin{prop}\label{prop}
        Consider the two optimisation programs %\footnote{The minimum snap objective applies only to $\pi_i(t;p_i)=x_i(t),~y_i(t)$ and $z_i(t)$. For $\pi_i(t;p_i)=\psi_i(t)$, the integrand in objective is $[\frac{d^2p_i}{dt^2}(t)]^2$ (see \cite{Mellinger2011}).}
        \begin{gather}
            \min_{\pi_i(t;p_i)} \quad \int_{t_0}^{t_1}\Big[\frac{d^4\pi_i(t;p_i)}{dt^4}\Big]^2 dt, \label{prob2} \\
            \text{s.t. } \quad \frac{d^j\pi_i(t;p_i)}{dt^j}=\delta_{i, j,t},\quad j=0,\ldots,s-1,~t=t_{i-1},t_i, \nonumber
        \intertext{and}
        	\min_{\overline{\pi}_i(t; \overline{p}_i)} \quad \int_{-1}^{1}\Big[\frac{d^4\overline{\pi}_i(t;\overline{p}_i)}{dt^4}\Big]^2 dt, \label{prob1} \\
        	\text{s.t. } \quad \frac{d^j\overline{\pi}_i(t;\overline{p}_i)}{dt^j}=\overline{\delta}_{i,j,t},\quad j=0,\ldots,s-1,~t=-1,1. \nonumber
    	\end{gather}
    	% \begin{subequations}\label{prob2}
        %     \begin{align}
        %         \min_{\pi_i(t;p_i)} \quad & \int_{t_0}^{t_1}\Big[\frac{d^4\pi_i(t;p_i)}{dt^4}\Big]^2 dt, \\
        %         \text{s.t. } \quad & \frac{d^j\pi_i(t;p_i)}{dt^j}=\delta_{i, j,t},\quad j=0,\ldots,s-1,~t=t_{i-1},t_i,
        %     \end{align}
        % \end{subequations}
        % and
        % \begin{subequations}\label{prob1}
        % 	\begin{align}
        % 	\min_{\overline{\pi}_i(t; \overline{p}_i)} \quad & \int_{-1}^{1}\Big[\frac{d^4\overline{\pi}_i(t;\overline{p}_i)}{dt^4}\Big]^2 dt, \\
        % 	\text{s.t. } \quad & \frac{d^j\overline{\pi}_i(t;\overline{p}_i)}{dt^j}=\overline{\delta}_{i,j,t},\quad j=0,\ldots,s-1,~t=-1,1.
        % 	\end{align}
        % \end{subequations}
        
        The polynomial segment $\overline{\pi}^\star_i(t;\overline{p}_i)$ solves \eqref{prob1} and $\pi^\star_i(t;p_i)=\overline{\pi}^\star_i(\Delta_i t + \Delta_i + t_{i-1}; \overline{p}_i)$ solves \eqref{prob2}, if
        \begin{subequations}\label{scaling_constraints}
            \begin{align}
                \Delta_i^{-j}\delta_{i,j,1} &= \overline{\delta}_{i,j,t_1},\hspace{1em}j=0,\ldots,s-1\\
                \Delta_i^{-j}\delta_{i,j,-1} &=\overline{\delta}_{i,j,t_0},\hspace{1em}j=0,\ldots,s-1.
            \end{align}
        \end{subequations}
    \end{prop}
    \vspace{1em}
    \begin{proof}
        Under the change of variable $\tau=\Delta_i^{-1}(t-\Delta_i-t_{i-1})$, let $\overline{\pi}_i(\tau; \overline{p}_i)=\pi_i(\Delta_i^{-1}(\tau-\Delta_i-t_{i-1}); p_i)$. With this change of variable, \eqref{prob2} becomes % for $j= 0,\ldots,s-1$
        \begin{gather}\label{prob3}
            \min_{\overline{\pi}_i(\tau; \overline{p}_i)} \quad \Delta_i^{-7}\int_{-1}^{1}\Big[\frac{d^4\overline{\pi}_i(\tau; \overline{p}_i)}{d\tau^4}\Big]^2 d\tau, \\
            \text{s.t. } \frac{d^j\overline{\pi}_i(\tau; \overline{p}_i)}{dt^j} =\Delta_i^{-j}\delta_{i,j,\tau},~j= 0,\ldots,s-1, \tau=-1,1. \nonumber
        \end{gather}
        % \begin{subequations}\label{prob3}
            % \begin{align}
            %     \min_{\overline{\pi}_i(\tau; \overline{p}_i)} \quad & \Delta_i^{-7}\int_{-1}^{1}\Big[\frac{d^4\overline{\pi}_i(\tau; \overline{p}_i)}{d\tau^4}\Big]^2 d\tau, \\
	           % \text{s.t. } \quad &  \frac{d^j\overline{\pi}_i(\tau; \overline{p}_i)}{dt^j} =\big\Delta_i^{-j}\delta_{i,j,\tau}, \quad j= 0,\ldots,s-1, \tau=-1,1.
            % \end{align}
        % \end{subequations}
        The minimization of \eqref{prob3} is independent of the constant factor. Therefore, if the constraints \eqref{scaling_constraints} are satisfied then \eqref{prob3} is equivalent to \eqref{prob1}.  If $\overline{\pi}_i(t;\overline{p}_i)$ solves \eqref{prob1} then $\overline{\pi}_i(\Delta_i t + \Delta_i + t_{i-1};\overline{p}_i)$ is the solution to \eqref{prob2}.
    \end{proof}
    
    Now we state a relationship that allows \eqref{form3} to be used to find a solution to \eqref{form1}.
    
    \begin{lemm}\label{scaled_shifted}
        Let $\overline{p}$ and $d$ solve \eqref{form3}, $p$ and $d$ solve \eqref{form1} and $\overline{\pi}_i(t;\overline{p}_i)$ and $\pi(t;p)$ be splines. Then
        \begin{align*}
        \pi(t;p)&=\begin{cases}
        \overline{\pi}_1(\Delta_1 t + \Delta_1 + t_0; \overline{p}_1) & t_0\leq t<t_1, \\
        \quad \vdots & \\
        \overline{\pi}_k(\Delta_k t + \Delta_k + t_{k-1}; \overline{p}_k) & t_{k-1}\leq t<t_k.
        \end{cases}
    \end{align*}
    \end{lemm}
    \vspace{1em}
    \begin{proof}
        Temporarily omitting the constraints \eqref{form1:con2} and \eqref{form1:con3}, the program in the nondimensional form is
        % 	We formulate the program with every segment in the nondimensional form \eqref{prob1}, temporarily omitting the constraints \eqref{form1:con2} and \eqref{form1:con3}
    	\begin{subequations} \label{form4}
    		\begin{gather}
    		\min_{\overline{p},d} \quad \overline{p}^T\overline{Q}\overline{p},  \\
    		\text{s.t.} \quad \overline{A}'\overline{p}=\overline{d}, \label{form4:con1}
    		\end{gather}
    		\end{subequations}
        where $\overline{A}'=\mathrm{diag}\{A_{\pm1},\ldots,A_{\pm1}\}\in\mathbb{R}^{kn\times kn}$ is constructed from elements $A_{\pm1}\in\mathbb{R}^{n\times n}$. We will next substitute \eqref{scaling_constraints} into \eqref{form4} to satisfy Proposition \ref{prop}. Stacking the derivatives $d_{i,j,t}$ and $\overline{d}_{i,j,t}$ for $j=0, \ldots,s-1$, the conditions \eqref{scaling_constraints} can be expressed as the matrix equation $\Gamma_i\overline{d}_i= d_i$.
    % 	\begin{align}\label{scaling_constraints2}
    % 	\Gamma_i\overline{d}_i&= d_i.
    % 	\end{align}
    	Substituting this expression into \eqref{form4:con1} yields \eqref{form3:con1}. We introduce the constraints \eqref{form3:con2} and \eqref{form3:con3} to constrain the dimensional derivatives as desired. The conditions \eqref{scaling_constraints} of Proposition \ref{prop} are satisfied by construction, hence the solution to \eqref{form3}, $\overline{p}^*(t)$, can be used to calculated the solution to \eqref{form1} using Lemma \ref{scaled_shifted}.
    \end{proof}
    
    \begin{rema}\label{rem:synergy}
        Replacing the original QP \eqref{form1} with the better conditioned \eqref{form3} in Problem \ref{prog1}, Algorithm \ref{algo1} can still be used for solution. With $A_i=\Gamma_i A_{\pm 1}$, $Q_i=Q_{\pm 1}$ and $\Sigma_i$, $b_{i,t_{i-1}}$ and $b_{i,t_i}$ as defined, Algorithm \ref{algo1} calculates $\overline{p}$. Then  the solution to Problem \ref{prog1}, $\pi(t)$, can be calculated using \eqref{scaled_shifted}.
    \end{rema}
    
    \begin{table}[t]
        \begin{center}
            \caption{Comparing the condition number of the transposed confluent Vandermonde matrices of \eqref{form1} and \eqref{form3} for two second segments with different times $t_{i-1}$ and $t_i$.}
            \begin{threeparttable}
                \begin{tabular}{rlcc}
                    \hline
                    &   & $\kappa(A_i)$ & $\kappa(\Gamma_i A_{\pm 1})$ \\
                    \hline
                    &$t_{i-1}=-1,~t_i=1$  & $2.3417 \cdot 10^{4}$ & $2.3417 \cdot 10^{4}$ \\
                    &$t_{i-1}=1,~t_i=3$ & $1.1775\cdot 10^{9}$ & $2.3417 \cdot 10^{4}$ \\
                    &$t_{i-1}=10,~t_i=12$  & $2.3763\cdot 10^{17}$ & $2.3417 \cdot 10^{4}$ \\
                    &$t_{i-1}=100,~t_i=102$ & $6.1296\cdot 10^{28}$ & $2.3417 \cdot 10^{4}$ \\
                    \hline
                \end{tabular}
            % \begin{tablenotes}
            %     \item[a] for the abstraction reaction, $\fam0 Mu+HX \rightarrow MuH+X$.
            %     \item[b] 1 degree${} = \pi/180$ radians.
            % \end{tablenotes}
            \end{threeparttable}
        \end{center}
        \label{conditioning}
    \end{table}
    
    \begin{figure}[t]
    	\centering
    	{\includegraphics[width=0.98\linewidth]{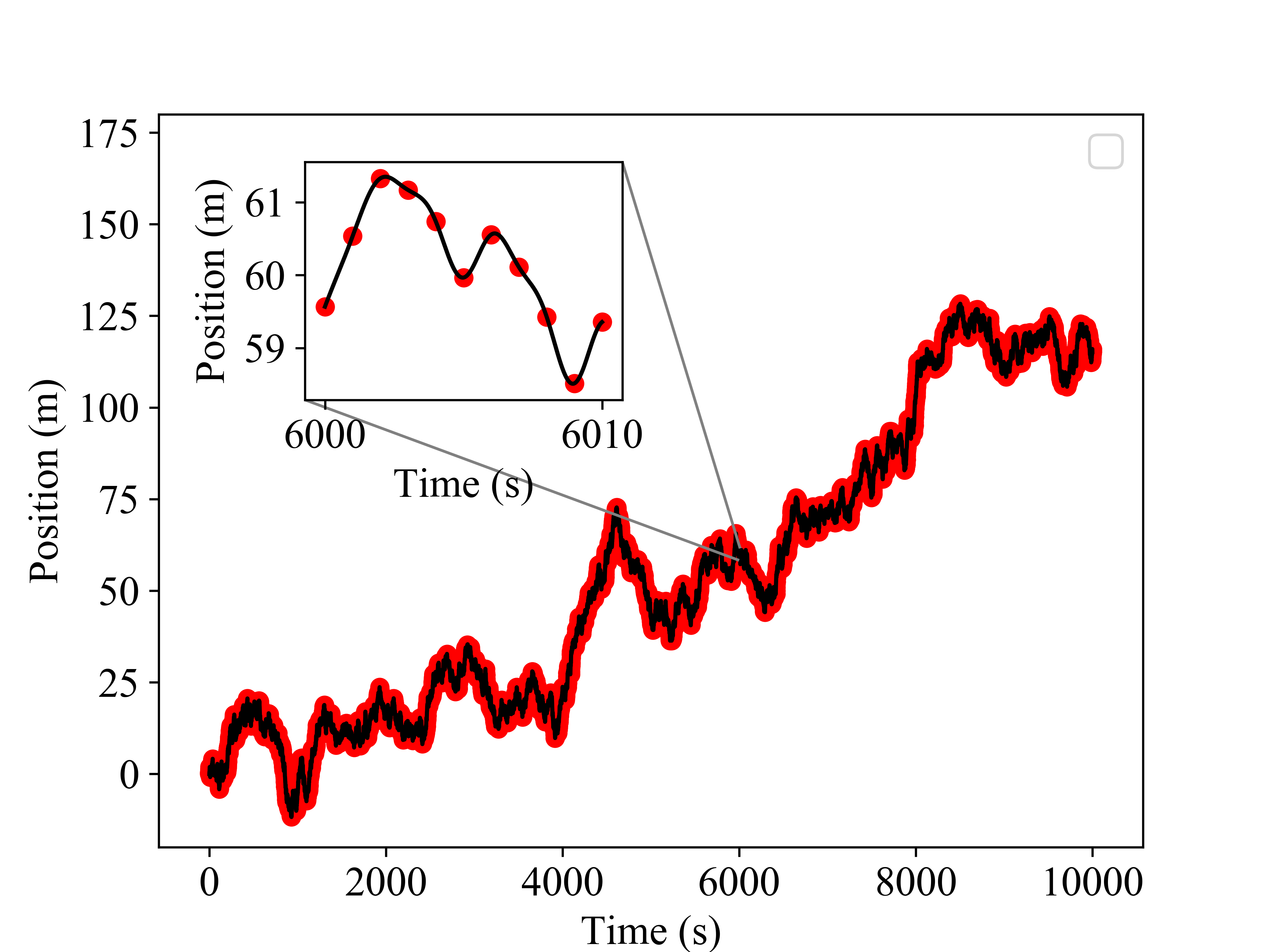}}
    	\caption{A spline of $10000$ segments generated by Algorithm \ref{algo1} with the nondimensional program \eqref{form3}. The problem was parameterized with positions from a random walk where each step $k$ is taken at time $k$ seconds and is of length taken uniformly from $[-1,1]$ in meters. The parametrized positions are plotted as red dots and the spline as the solid black line.}
    	\label{random_walk} 
    \end{figure}
    
    The improvement of the condition number of the reformulated $\Gamma_iA_{\pm1}$ compared to the original $A_i$ is demonstrated by Table \ref{conditioning}. The condition number $\kappa(A_i)$ rapidly grows to more than $10^{17}$ for $t_{i-1}=10$ and $t_i=12$, where the matrices become practically unusable in calculations given numerical errors \cite{trefethen1997numerical}. % If Problem \ref{prog1} is parametrized with a flight duration greater than $t_k=10$, the constraint matrices $A_i$ will exhibit condition numbers of similar magnitudes. 
    This helps in explaining the instability observed in generating large trajectories, as the $t_i$ increase with the number of segments. We also note the $\kappa(\Gamma_iA_{\pm1})$ remain constant in Table \ref{conditioning}, due to constant $\Delta_i=1$. The $t_{i-1}$ and $t_i$ only appear in $\Gamma_iA_{\pm1}$ as part of the difference terms $\Delta_i$. Hence, $\kappa(\Gamma_iA_{\pm1})$ does not increase with $t_{i-1}$ and is only prone to introducing numerical error for large $\Delta_i$.
    
    \section{EXPERIMENTAL RESULTS}\label{sec::experiments}
    
    \begin{figure}[tp]
    	\centering
    	{\includegraphics[width=0.7\linewidth]{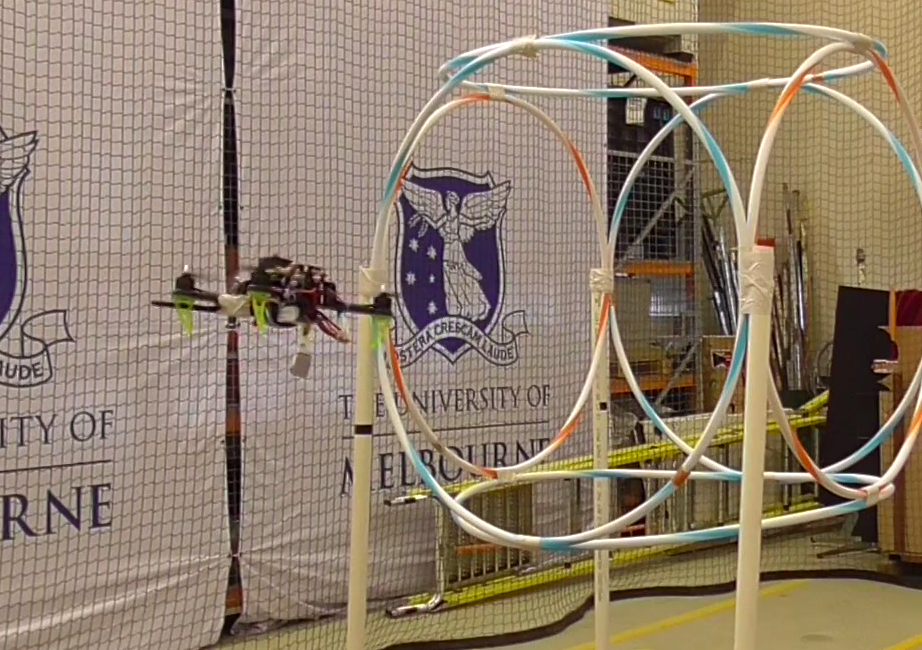}}
    	\caption{Snapshot of the quadrotor flying through hoops as part of a long trajectory. The total flight time was more than $2$ min and the trajectory was calculated offline in $0.3319$ s.}
    	\label{still} 
    \end{figure}
    
    \begin{figure}[tp]
    	\centering
    	{\includegraphics[width=0.95\linewidth]{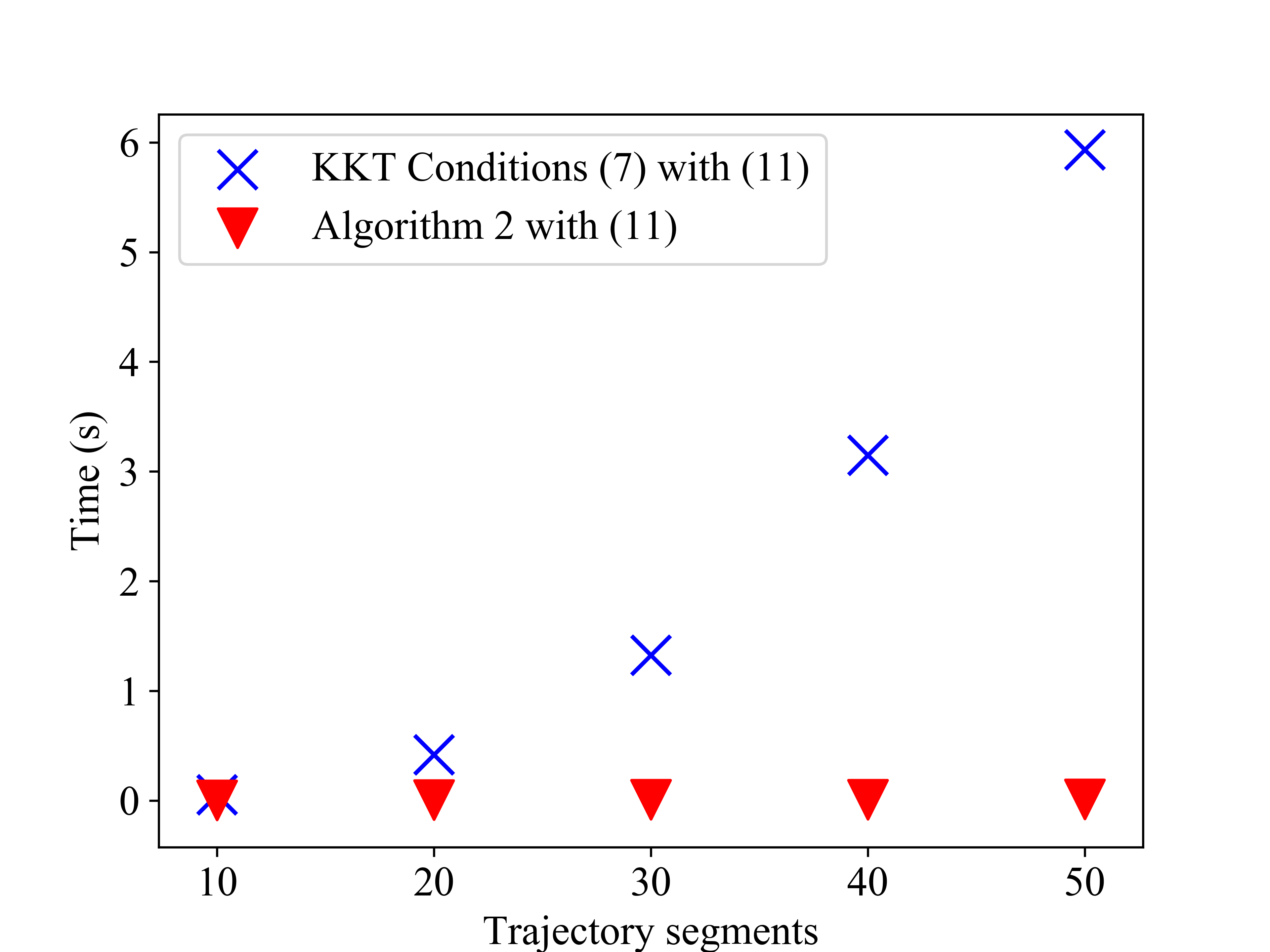}}
    	\caption{The computation time required to generate trajectories for a range of $k$ segments of Algorithm \ref{algo1} (triangles) and by solving the equations from the KKT conditions \eqref{KKT} (crosses).}
    	\label{compcomp} 
    \end{figure}
    
    % \begin{figure}[tp]
    % 	\centering
    % 	{\includegraphics[width=0.98\linewidth]{tracking.png}}
    % 	\caption{The quadrotor's position in $x$ for a $2$ min long trajectory. }
    % 	\label{tracking} 
    % \end{figure}
    
    \subsection{Large Scale Trajectory Generation}
    
    We demonstrated the performance of Algorithm \ref{algo1} with the nondimensional program \eqref{form3} by generating a minimum-snap trajectory of a large number of segments. Implementing Algorithm \ref{algo1} with \eqref{form3} in C++, we were able to generate a spline of more than $500000$ segments in $156.01$ s. The algorithm was implemented and run on a laptop with an Intel Core i7-8650U CPU running at 1.9 GHz, with 16 GB of RAM. The limitation of algorithm was memory allocation on the computer, and the implementation was not optimized to increase performance. 
    
    \subsection{Long Quadrotor Flight}
    
    We conducted an experiment to demonstrate the feasibility of the minimum-snap trajectories generated by Algorithm \ref{algo1} with \eqref{form3}. A trajectory was generated such that a quadrotor repeated excursions from the centre of a virtual cube, which was represented in the flight space with hoops as pictured in Fig. \ref{still}. The trajectory had more than $80$ segments with a flight time greater than $2$ min. A spline was calculated offline in $0.3319$ s and used as a position reference, evaluated using a ground station computer's clock and transmitted to the vehicle at $20$ Hz.
    
    For the experimental testbed, we flew a quadrotor with generic components and a Pixhawk2 The Cube Black \cite{PX4DevTeam2019} flight controller. We use a Vicon motion capture system \cite{ViconMotionSystemsLtdUK2020} and the flight controller's onboard sensors for sensing. We leveraged the PX4 firmware \cite{PX4DevTeam2020} for control and estimation.
    
    % For the problem parameterization, the trajectory is specified as a repeated excursion from the centre of a virtual cube, through a random face of the cube and back to the center. Hoops were used in the flight space to represent the faces of the virtual cube, as pictured in Fig. \ref{still}. An example of tracking performance of the minimum-snap trajectory is included in Fig. \ref{tracking}.
    
    \subsection{Computation Time Comparison}
    
    Finally, we performed a series of benchmark tests to demonstrate the performance of Algorithm \ref{algo1} with \eqref{form3} compared with solving the equations from the KKT conditions \eqref{KKT}. This method was chosen for comparison given it is of similar computational complexity in $k$ as the method proposed by Richter {\it et al.} \cite{Richter2016}. The reformulated optimization program \eqref{form3} was used instead of the original program \eqref{form1} in order to compare performance in calculating a range of $k$ segments without numerical instability. Fig. \ref{compcomp} shows the linear and cubic computation time of Algorithm \ref{algo1} and the solution of \eqref{KKT}, respectively. The computation time of Algorithm \ref{algo1} ranges from $2.56$ to $12.41$ ms for calculating splines of $10$ to $50$ segments, whereas it takes $5.93$ s to calculate a $50$ segment spline by solving \eqref{KKT}.

    \section{Conclusions and Future Works}\label{sec::conclusion}
    
    In this paper we have presented a new algorithm to generate minimum-snap spline trajectories for quadrotors in linear computational complexity in the number of segments. We also proposed a reformulation of the associated optimization program with better conditioned matrices than the original program. With these two developments, we are able to generate large trajectories not limited by computational time or number of segments but instead by computer memory. The performance of the proposed algorithm furthers the applicability of minimum-snap trajectory planning for real-time applications, enabling the faster calculation of larger trajectories onboard quadrotors.
    
    In a sequel, we will explore further applications enabled by the algorithm's computational complexity. For example, the algorithm can be used as a fast oracle for zeroth-order optimization algorithms. Such approaches have been suggested to minimize the snap of a trajectory by optimizing the time parameterization $T$. % Another application is in the solution of mixed-integer quadratic programs, where an efficient method of solving underlying QPs significantly improves the computation time of branch and bound methods.  
    
    \appendix{}
    
    \subsection{Constructing Hessian Matrices}\label{sec:const}
    There are two possible constructions of the Hessian matrices $Q_i(T)\in\mathbb{R}^{nk\times nk}$ in \eqref{form1}, corresponding to $\pi_i(t;p_i)=x_i(t),~y_i(t)$ or $z_i(t)$, and $\pi_i(t;p_i)=\psi_i(t)$. Both are constructed by calculating the quadratic form that results from the following integral in representing $\pi_i(t;p_i)$ as a vector of coefficients $p_i$
    \begin{align*}
        \int_{t_0}^{t_1}\frac{d^{s-1}\pi_i(\tau;p_i)}{d\tau^{s-1}}^2d\tau&=p_i^TQ_i(T)p_i,
    \end{align*}
    where $s=5$ for $\pi_i(t;p_i)=x_i(t),~y_i(t)$ or $z_i(t)$, such that the integrand is the square of the snap of position, or $s=3$ for $\pi_i(t;p_i)=\psi_i(t)$, such that the integrand is the square of the acceleration of position.

    \subsection{Proof of Proposition 3}\label{sec:proof}
    Before stating the proof, we first formulate \eqref{form1} so that the desirable structure of the matrices in the program is revealed. From this reformulation, we obtain an efficient method in a similar fashion as Cantoni {\it et al.} \cite{cantoni2020structured}, from which we derive the steps of Algorithm \ref{algo1}.
    % \begin{figure*}[!t]
    % 	\normalsize
    % 	%        \setcounter{MYtempeqncnt}{\value{equation}}
    % 	% Set the equation number to one less than the one
    % 	% desired for the first equation here.
    % 	% The value here will have to changed if equations
    % 	% are added or removed prior to the place these
    % 	% equations are referenced in the main text.
    % 	% \setcounter{equation}{5}
    	
    % 	\begin{align}
    % 	\begin{bmatrix}
    % 	R_1 &  -\Psi^T & & & \\
    % 	-\Psi & R_2 & -\Psi^T & & \\
    % 	& -\Psi & R_3& -\Psi^T  &  \\
    % 	& & & \ddots & \\
    % 	& & & -\Psi & R_k
    % 	\end{bmatrix}&=\begin{bmatrix}
    % 	I &   & &  & \\
    % 	\Phi_1 & I & & & \\
    % 	& \Phi_2 & I & & \\
    % 	& & & \ddots & & \\
    % 	& & & \Phi_{k-1} & I
    % 	\end{bmatrix}\begin{bmatrix}
    % 	R_1 &  -\Psi^T & & & \\
    % 	& \overline{R}_2 & -\Psi^T & & \\
    % 	& & \overline{R}_3 & -\Psi^T & \\
    % 	& & & \ddots & & \\
    % 	& & & & \overline{R}_k
    % 	\end{bmatrix} \label{LU}
    % 	\end{align}
    % 	% % Restore the current equation number.
    % 	% \setcounter{equation}{\value{MYtempeqncnt}}
    % 	% IEEE uses as a separator
    % 	% \vspace{1em}
    % 	\hrulefill
    % 	% The spacer can be tweaked to stop underfull vboxes.
    % 	% \vspace*{1em}
    % \end{figure*}
    
    % \subsubsection{Reformulating the Problem}
    Similar to Richter {\it et al.} \cite{Richter2016}, we reformulate \eqref{form1} by replacing \eqref{form1:con1} with $p=A^{-1}d$. % as the following.
    % \begin{subequations}\label{form2}
    %     \begin{align}
    %         \min_{d} \quad &d^TA^{-T}QA^{-1}d, \label{cost}\\
    %         \text{s.t. } \quad & \Lambda d =0, \label{coupling} \\
    %         \quad & \Pi d = \beta. \label{fixed}
    %     \end{align}
    % \end{subequations}
    We further replace \eqref{form1:con3} with \eqref{proj1}. By construction $\Pi d=\Pi (b +  \Sigma f )= \beta$, so we substitute \eqref{proj1} into \eqref{form1} and are left with
    \begin{subequations}\label{reform2}
        \begin{align}
        \min_{f} \quad &f^T\Sigma^TA^{-T}QA^{-1}b + \frac{1}{2}f^T\Sigma^TA^{-T}QA^{-1}\Sigma f, \label{cost2} \\
        \text{s.t. } \quad & \Lambda \Sigma f =0.
    \end{align}
    \end{subequations}
    
    % Cast as \eqref{reform2}, the Hessian $\Sigma^TA^{-T}QA^{-1}\Sigma$ is block diagonal and the decision variables are coupled by the sparse matrix $\Lambda \Sigma$. Furthermore, the coupling is only between the derivatives of adjacent segments. Motivated by these observations, we now derive an efficient method that exploits this structure rather than using dense computations. This method is used to derive the calculations in Algorithm \ref{algo1}.
    
    % \subsubsection{A Structured Solution of the KKT}
    The KKT conditions for \eqref{reform2} yield
    \begin{align}\label{KKT}
        \begin{bmatrix}
        \Sigma^TA^{-T}QA^{-1}\Sigma & \Sigma^T\Lambda^T \\
        \Lambda\Sigma & 0
        \end{bmatrix}\begin{bmatrix}
        f \\ \lambda
        \end{bmatrix} &= \begin{bmatrix}
       \Sigma^TA^{-T}QA^{-1}b \\ 0
        \end{bmatrix},
    \end{align}
    where the Lagrangian multipliers are $\lambda = [\lambda_1^T, \dots, \lambda_{k-1}^T]^T$. % and $\lambda_i\in \mathbb{R}^s$. % We partition $g$ the same as $f$, $g=[g_{1,0}^T,g_{1,T}^T, \ldots,g_{k,0}^T,g_{k,T}^T]^T$.
 
    Choosing the partition of variables as \eqref{partition}, we may permute the variables and columns of \eqref{KKT} to reveal the block-tridiagonal structure
    \begin{align}\label{tridiag}
        \underbrace{ \begin{bmatrix}
        R_1 &  -\Psi^T & & \\
        -\Psi & R_2 & -\Psi^T & \\
        & -\Psi & R_3 & \hdots\\
        & & \vdots & \ddots
        \end{bmatrix}}_{R}\begin{bmatrix}
        y_1 \\ y_2 \\ y_3 \\ \vdots 
        \end{bmatrix} &= \begin{bmatrix}
        c_1 \\ c_2 \\ c_3 \\ \vdots 
        \end{bmatrix}, 
        \end{align}
        where 
        \begin{align*}
            R_i=\begin{bmatrix}
            B_i & C_i & 0 \\
            C_i^T & D_i & I \\
            0 & I & 0
            \end{bmatrix}, \quad
            \Psi= \begin{bmatrix}
            0 & 0 & I \\
            0 & 0 & 0 \\
            0 & 0 & 0
            \end{bmatrix},
    \end{align*}
    and $y_i=[f_{i,t_{i-1}}^T, f_{i,t_i}^T, \lambda_i^T]^T$ and $c_i=[g_{i,t_{i-1}}^T, g_{i,t_i}^T, 0]^T$.
    
    Block-tridiagonal matrices such as \eqref{tridiag} are commonly solved through Block LU factorisation \cite{Fox1969}. Derived from this factorisation, the following set of recursions solve \eqref{KKT}. In satisfying the optimality conditions for \eqref{form1}, the following is our efficient method for solving the optimization program in Problem \ref{prog1}.
    
    \begin{lemm}\label{structured_sol}
    For $i=1\ldots,k$, let $f_{i,t_{i-1}},f_{i,t_{i}}$ and $\lambda_i$ satisfy 
    \vspace{-1em}
        \begin{subequations}\label{recs}
            \begin{align}
            f_{i,t_{i-1}}
            &=
            \overline{B}^{-1}_i(\overline{g}_{i,t_{i-1}}-C_if_{i,t_{i}}), \label{rec1} \\
            f_{i,t_{i}}
            &=\begin{cases}
            \begin{aligned}
            & (D_k-C_k^T\overline{B}_k^{-1}C_k)^{-1}  \\
            & \quad (g_{k,t_{k}}-C_k^T\overline{B}_k^{-1}\overline{g}_{k,t_{k-1}})
            \end{aligned} 
            & i=k, \\
            f_{i+1,t_{i}} & i=k-1, \ldots,1,
            \end{cases} \label{rec2} \\
             g_{i,t_{i}} &= C_i^T f_{i,t_{i-1}} + D_i f_{i,t_i} + \lambda_i, \label{rec3}
            \end{align}
        \end{subequations}
        where
        \begin{subequations}
            \begin{align}
            \overline{B}_i &=\begin{cases}
            B_1 & i=1, \\
            B_{i}+D_{i-1}-C_{i-1}^TB_{i-1}^{-1}C_{i-1} & i=2,\ldots,k,
            \end{cases} \label{con1}\\
            \overline{g}_{i,t_{i-1}}&=\begin{cases}
            {g}_{1,t_{i-1}} & i=1, \\
            \begin{aligned}
            & g_{i,t_{i-1}} + g_{i-1,t_{i-1}} \\
            & \quad -C_{i-1}^TB_{i-1}^{-1}\overline{g}_{i-1,t_{i-1}}
            \end{aligned}
             & i=2,\ldots,k.
            \end{cases} \label{con2}
            \end{align}
        \end{subequations}
        Then \eqref{recs} solves \eqref{KKT}.
    \end{lemm}
    % \vspace{1em}
    \begin{proof}
        Let the LU factorisation of the block tridiagonal matrix $R$ in \eqref{tridiag} be $R=LU$ as
        \begin{align*}
    	R &=\begin{bmatrix}
    	I &   & & \\
    	\Phi_1 & I & & \\
    	& \Phi_2 & & \\
    	& & \ddots & & \\
    	& & \Phi_{k-1} & I
    	\end{bmatrix}\begin{bmatrix}
    	R_1 & -\Psi^T & & \\
    	& \overline{R}_2 & -\Psi^T & \\
    	& & \overline{R}_3 & \\
    	& & \ddots & & \\
    	& & & \overline{R}_k
    	\end{bmatrix},
    	\end{align*}
        with the block elements for $i=1,\ldots,k-1$
        \begin{align*}
            \Phi_i&=\begin{bmatrix}
            C_i^TB_i^{-1} & -I  & D_i-C_i^TB_i^{-1}C_i \\
            0 & 0 & 0  \\
            0 & 0 & 0  
            \end{bmatrix},
        \intertext{and for $i=2,\ldots,k$}
            \overline{R}_i&=\begin{bmatrix}
            B_i+D_{i-1}-C_{i-1}^TB_{i-1}^{-1}C_{i-1} & C_i & 0 \\
            C_i^T & D_i & I \\
            0 & I & 0
            \end{bmatrix}.
        \end{align*}
    Solving $Lx=c$ yields the iteration \eqref{con2}. Calculating each $B_i+D_{i-1}-C_{i-1}^TB_{i-1}^{-1}C_{i-1}$ gives rise to \eqref{con1}. These matrices are used to solve $Uy=x$, providing expressions for $f_{i,t_{i-1}}$ and $f_{i,t_{i}}$ as \eqref{rec1} and \eqref{rec2}. The solution to $Uy=x$ also governs the values of $\lambda_i$ with \eqref{rec3}.
    \end{proof}
    Lemma \ref{structured_sol} can be used to solve Problem \ref{prog1} with linear computational complexity in $k$.  There are $2k$ matrix calculations required in computing \eqref{con1} and \eqref{con2}, while $2k$ systems of equations need to be solved in \eqref{rec1} and \eqref{rec2}. All the matrices involved are square and at largest $s$, and the computational complexity is then $O(\frac{4}{3}ks^3)$.

    \bibliographystyle{IEEEtran}
    \bibliography{root}

\end{document}